\documentclass[12pt,reqno]{article}

\usepackage{setspace}
\usepackage{latexsym}
\usepackage{amsthm}
\usepackage{amssymb}
\usepackage{amsmath}
\usepackage{amsfonts,amsgen,amsopn}
\usepackage{authblk}

\usepackage{epsfig} \usepackage{color}
\usepackage[font=scriptsize,labelfont=bf]{caption} 

\usepackage{subfig}

\usepackage[a4paper,
hdivide={2.5cm,,2.5cm},vdivide={3.5cm,,2cm}]{geometry}

\usepackage{natbib}
  
\usepackage[english]{babel}
\usepackage{myheadings}
\usepackage{fancyheadings}
\newtheorem{theorem}{Theorem}[section]
\newtheorem{lemma}[theorem]{Lemma}

\newtheorem{corollary}[theorem]{Corollary}
\newtheorem{proposition}[theorem]{Proposition}

\theoremstyle{remark}

\newtheorem{remark}{Remark}
\newcommand{\T}{{\mathcal T}}
\newcommand{\C}{{\mathcal C}}
\newcommand{\D}{{\mathcal D}}

\newcommand{\NN}{{\mathbb N}}

\begin{document}

\bibliographystyle{plainnat}

\begin{titlepage}
  \title{\bfseries Revisiting an equivalence between maximum parsimony
    and maximum likelihood methods in phylogenetics}
  \pagestyle{headings} \markboth{Equivalence of MP and ML}{Equivalence
    of MP and ML} \date{}

  \author{\normalsize Mareike Fischer} \affil{Corresponding author}
  \affil{\footnotesize Allan Wilson Centre for Molecular Ecology and
    Evolution, Biomathematics Research Centre, University of Canterbury,
    New Zealand; E-mail: email@mareikefischer.de}

  \setcounter{Maxaffil}{1}

  \author{\normalsize Bhalchandra Thatte} \affil{\footnotesize
    Department of Statistics, University of Oxford, United Kingdom;
    E-mail: thatte@stats.oxford.ac.uk}

\end{titlepage}
\maketitle

\begin{abstract}
  Tuffley and Steel (1997) proved that Maximum Likelihood and Maximum
  Parsimony methods in phylogenetics are equivalent for sequences of
  characters under a simple symmetric model of substitution with no
  common mechanism. This result has been widely cited ever since. We
  show that small changes to the model assumptions suffice to make the
  two methods inequivalent. In particular, we analyze the case of
  bounded substitution probabilities as well as the molecular clock
  assumption. We show that in these cases, even under no common
  mechanism, Maximum Parsimony and Maximum Likelihood might make
  conflicting choices. We also show that if there is an upper bound on
  the substitution probabilities which is `sufficiently small', every
  Maximum Likelihood tree is also a Maximum Parsimony tree (but not vice
  versa). \par
  \vspace{0.3cm} {\noindent \bfseries Keywords:} phylogenetics, maximum
  parsimony, maximum likelihood, molecular clock
\end{abstract}

\section{Introduction}
Stochastic models for nucleotide substitution and tree reconstruction
methods for inferring phylogenetic trees are used to interpret
the ever-growing amount of available genetic sequence
data. Unsurprisingly, such models and methods have therefore been widely
discussed in the last decades (e.g., \citep{felsenstein_1978};
\citep{felsenstein_2004}; \citep{semple_steel_2003};
\citep{yang_2006}). Two of the most frequently used tree reconstruction
methods are {\em Maximum Parsimony} (MP) and {\em Maximum Likelihood}
(ML). A basic difference between these two methods is that MP, unlike
ML, is not based on a specific nucleotide substitution model. If the
sequences under consideration are related by a specific model of
substitution, the results of MP and ML may coincide
\citep{hendy_penny_1989}, but there are also examples, such as the
famous `Felsenstein Zone', for which this is not the case
\citep{felsenstein_1978}.

In 1997, Tuffley and Steel took an important step forward in the
analysis of MP and ML \citep{tuffley_steel_1997}: they showed that a
particular symmetric model of substitution with `no common mechanism' is
sufficient for MP and ML to be equivalent when applied to a sequence of
characters.

The purpose of this paper is to analyze this equivalence of MP and ML
further by considering slightly modified model assumptions that are of
biological relevance. For instance, MP is often assumed to be justified
whenever the nucleotide substitution proba\-bilities are small (e.g.,
\citep{felsenstein_2004}, p. 101). Therefore, we restrict the model by
placing an upper bound on these probabilities, and find that under no
common mechanism MP and ML are no longer equivalent. Moreover, the
equivalence of MP and ML under a `no common mechanism model' also fails
under the constraint of a mole\-cular clock, even without a bound on the
substitution probabilities. These two claims will be established by
constructing counterexamples that are minimal with respect to the number
of taxa. To construct our examples, we exploit a useful property of the
likelihood function for a `no common mechanism' model, namely that it is
multilinear in the substitution probabilities. This fact underlies
Equation 18 and Lemma 2 in \citep{tuffley_steel_1997}, which we use in
our arguments.

We then go on to prove bounds on the probability of observing a given
sequence of characters on a tree, and use them to show that it is
possible to choose sufficiently small substitution probabilities
(depending on the number of taxa, the number of characters and the
number of states) so that every tree chosen by ML is also a most
parsimonious tree.

\section{Notation and Model Assumptions}
\label{notation}

Recall that a {\it phylogenetic $X$-tree} is a tree $\T =(V(\T),E(\T))$
on a leaf set $X=\{1,\ldots,m\} \subset V(\T)$ with no vertices of
degree $2$. Note that the tree does not have to be binary. Furthermore,
recall that a {\it character} $f$ is a function $f: X\rightarrow \C$ for
some set $\C:=\{c_1, c_2, c_3, \ldots, c_r \}$ of $r$ {\em character
  states} ($r \in \NN$). An {\it extension} of $f$ to $V(\T)$ is a map
$g: V(\T)\rightarrow \C$ such that $g(i) = f(i)$ for all $i$ in $X$. For
such an extension $g$ of $f$, we denote by $l_{\T}(g)$ the number of
edges $e=\{u,v\}$ in $\T$ on which a substitution (mutation) occurs,
i.e. where $g(u) \neq g(v)$. The {\em parsimony score} of $f$ on $\T$,
denoted by $l_{\T}(f)$, is obtained by minimizing $l_{\T}(g)$ over all
possible extensions $g$. The parsimony score of a sequence of characters
$S:= f_1f_2\ldots f_n$ is given by
$l_{\T}(S)=\sum\limits_{i=1}^n\;l_{\T}(f_i)$.

Recall that a character $f$ on a leaf set $X$ is said to be {\em informative}
(with respect to parsimony) if at least two distinct character states
occur more than once on $X$. Otherwise $f$ is called {\em non-informative}. Note
that for a non-informative character $f$, $l_{\T_i}(f)=l_{\T_j}(f)$ for
all trees $\T_i$, $\T_j$ on the same set $X$ of leaves.
  
Next we describe the fully symmetric $r$-state model
\citep{neyman_1971}, also known as the $N_r$-model, which underlies the
Tuffley and Steel equivalence result.

Consider a phylogenetic $X$-tree $\T$ arbitrarily rooted at one of its
vertices. The $N_r$-model assumes that a state is assigned to the root
from the uniform distribution on the set of states. The state then
evolves away from the root as follows. The model assumes equal rates of
substitutions between any two distinct character states. For any edge $e
= \{u,v\} \in E(\T)$, where $u$ is the vertex closer to the root, let
$p_e$ denote the conditional probability $P(v=c_i|u=c_j)$, where $c_i
\neq c_j$. The probability $p_e$ is equal for all pairs of distinct
states $c_i$ and $c_j$. Therefore, the probability that a substitution
($c_j$ to a state different from $c_j$) occurs on the edge $e$ is
$(r-1)p_e$. Let $q_e$ be the conditional probability $P(v=c_i|u=c_i)$,
i.e. the probability that no substitution occurs on edge $e$. In the
$N_r$-model, we have $0\leq p_e\leq \frac{1}{r}$ for all $e \in E(\T)$,
and $(r-1)p_e + q_e = 1$. Moreover, the $N_r$-model assumes that
substitutions on different edges are independent. Note that for $r=4$,
the $N_r$-model coincides with the Jukes-Cantor model
\citep{jukes_cantor_1969}.

Let $\T$ be a phylogenetic $X$-tree and let $f$ be a character on its
leaf set $X$.  Let the substitution probabilities assigned to the edges
of $\T$ under the $N_r$-model be collectively denoted by $\bar{p}:=
(p_e: e\in E(\mathcal{T}))$. Then we denote by $P(f|\T, \bar{p})$ the
probability of observing character $f$ given tree $\T$ and the
parameter values $\bar{p}$. Note that $P(f|\T, \bar{p})$ does not depend
on the root position, since the model is symmetric. The maximum value of
this probability for fixed $f$ and $\T$ as $\bar{p}$ ranges over all
possibilities is denoted by $\max P(f|\T)$, i.e. $\max P(f|\T) :=
\max_{\bar{p}}P(f|\T,\bar{p})$.

Now let $S:=f_1,\ldots,f_n$ be a sequence of characters. In this paper, 
we analyze sequences of characters under the $N_r$-model with {\em no
  common mechanism}. This means that the substitution probabilities on
edges may be different for different characters in $S$ without any
correlation between the characters. We suppose that for each character
$f_i$ in the sequence and for each edge $e$ of the tree, there is a
parameter $p_{e,i}$ that gives the substitution probability for $f_i$ on
edge $e$, and that the parameters $p_{e,i}$ are all independent. For
$i=1,\ldots,n$, let $\bar{p}_i:= (p_{e,i}: e\in E(\mathcal{T}))$ be the
vectors of substitution probabilities. We denote the model parameters
$(\bar{p}_i, i=1,\ldots,n)$ collectively as $\Theta$. Then the
probability of observing the sequence of characters $S$ on tree $\T$ for
the given parameters $\Theta$ is given by:

\begin{equation*}
  P(S|\T, \Theta) = \prod_{i=1}^nP(f_i|\T, \bar{p}_i),
\end{equation*}
which follows from the fact that the characters are independent.

We refer to $P(S|\T, \Theta)$ as the probability of observing sequence $S$
given the phylogenetic tree $\T$ and model parameters $\Theta $. We then
define the likelihood of the tree $\T$ and the model parameters $\Theta $
given the sequence $S$, which we refer to as the {\em likelihood function}, as
$L(\T, \Theta|S) := P(S|\T, \Theta)$. The maximum likelihood method of
phylogenetic tree reconstruction involves optimizing the likelihood
function in two steps as described in \citep{semple_steel_2003}. We
first maximize $P(S|\T,\Theta)$ over the space of model parameters
$\Theta$. We define:
\begin{equation*}
  \max P(S|\T) :=  \max_{\Theta }P(S|\T,\Theta).
\end{equation*} 
We then choose a tree $\T$ that maximizes $\max P(S|\T)$. We call such a
tree a maximum likelihood tree (ML-tree) of $S$. Thus, an ML-tree of a
sequence $S$ is $\mathrm{argmax}_{\T}\left( \max
  P(S|\T)\right)$. Note that under the assumption of no common
mechanism, we have:
\[ \max P(S|\T) =
\prod_{i=1}^n\max_{\bar{p}_i}P(f_i|\T,\bar{p}_i).  \]

\section{Results}
Using the notation introduced in the previous section, we are now in a
position to state the equivalence result of Tuffley and Steel
explicitly.

\begin{theorem} {\bf (Tuffley and Steel 1997).}
  \label{thm-tuffley-steel}
  Let $\T$ be a phylogenetic $X$-tree and let $S:=f_1,\ldots,f_n$
  be a sequence of $r$-state characters on $X$. Then, under the
  $N_r$-model with no common mechanism, we have:
\begin{equation*}
  \label{eq_no_common_mechanism}
 \max P(S|\T) = r^{-l_{\T}(S)-n}.
\end{equation*} 
Thus ML and MP both choose the same tree(s).
\end{theorem}

In the following, we show that small changes to the assumptions of the
$N_r$-model may be enough to make this equivalence fail. In particular, we
analyze two settings of biological interest: first, we consider bounded
substitution probabilities; secondly we investigate the case of a
molecular clock. In both cases, we explicitly construct examples in which
MP and ML choose different sets of trees under no common mechanism.

\subsection{Bounded substitution probabilities}  
In this section, we consider a modification of the $N_r$-model in which
the substitution probabilities on all edges are bounded above by some $u
< \frac{1}{r}$. We construct character sequences for which MP and ML
choose different sets of trees.

{\begin{proposition}
    \label{no-clock} Under the $N_r$-model with no common
    mechanism, for $r \geq 2$, there exist values of $u$ such that if
    the substitution probabilities are bounded above by $u$, MP and ML
    choose different sets of trees.  In particular, we have:
  \begin{enumerate}
  \item For $r=2$, for all values of $u \in \left(0,1 -
      \frac{1}{\sqrt{2}}\right)$, there exist sequences of characters
    for which MP and ML choose different sets of trees.
  \item For $r > 2$, for all values of $u \in (0,\frac{1}{r})$ there
    exist sequences of characters for which MP and ML choose unique and
    distinct trees.
  \end{enumerate}
\end{proposition}

In order to prove this proposition, it is necessary to summarize
the main idea of the original proof of the Tuffley-Steel result. We
state it here in a more general form so that it may be used to analyze
the situation in which the substitution probabilities are bounded.

\begin{lemma}
  \label{corner}
  Let $\T$ be a phylogenetic $X$-tree and let $f$ be a character on
  $X$. Then under the $N_r$-model with all substitution probabilities
  bounded by $u$, where $0 \leq u \leq \frac{1}{r}$, the probability
  $P(f|\T, \bar{p})$ can be maximized at a point where all substitution
  probabilities are either 0 or $u$.
\end{lemma}

Lemma \ref{corner} is the same as Lemma 2 in \citep{tuffley_steel_1997}
except that Tuffley and Steel stated their result only for
$u=\frac{1}{r}$. However, this assumption is not used in their proof and is 
therefore not required for the lemma to hold. Tuffley and Steel used it
to explicitly maximize the probability of observing a character on a
given tree under the $N_r$-model: for a given character $f$ and tree
$\T$ with a most parsimonious extension $g$ of $f$, assigning
substitution probability $\frac{1}{r}$ to edges where a substitution is
induced by $g$, and $0$ elsewhere, gives $\max_{\bar{p}} P(f|\T,
\bar{p})$ (cf. Theorem 3 of \citep{tuffley_steel_1997}).

\par But it turns out that an ML solution cannot be similarly related to
an MP solution when $u < \frac{1}{r}$. That is, if $g$ is a most
parsimonious extension of a character $f$, then we may not be able to
maximize the probability by simply assigning the substitution
probability $u$ to edges on which there is a substitution in $g$, and 0
to edges on which there is no substitution in $g$. The probability may
actually be maximized at some other corner of the feasibility region of 
$\bar{p}$. This is the idea of the following construction.

\begin{proof}[Proof of Proposition ~\ref{no-clock}]
  We provide examples of sequences of characters for which MP and ML may
  choose different sets of trees. We first prove the case $r=2$ with an
  example on five taxa, and show that in this case, there are no such
  examples on fewer than five taxa. Then we explicitly prove the case $r=3$
  with an example on four taxa and show how this example can be generalized
  for $r > 3$.

  \noindent {\bf Case $r=2$:} 

  Let the set of character states be $\{\alpha,\beta \}$. Consider the
  two trees $\T_1$ and $\T_2$ shown in Figure \ref{bounded-fig-old}
  alongside the characters $f_1=\alpha\alpha\beta\beta\beta$ and
  $f_2=\alpha\beta\alpha\beta\beta$. We consider the character sequence
  $S:=f_1f_2$.

  \begin{figure}[ht]
    \centering\vspace{0.5cm}
    \includegraphics[width=15cm]{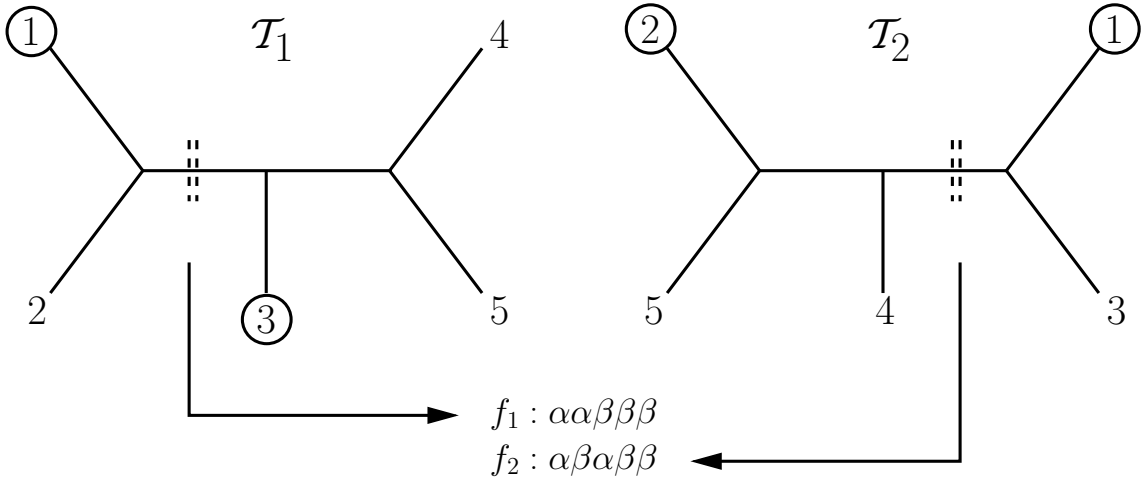} 
    \caption{The characters $f_1$ and $f_2$ both correspond to a split
      on an interior edge of $\T_1$ or $\T_2$, respectively. But, as
      highlighted by the circled leaves, the assignment of $f_1$ on
      $\T_2$ differs from the assignment of $f_2$ to
      $\T_1$.}  
    \label{bounded-fig-old} 
    \vspace{0.5cm} 
  \end{figure} 
  Note that $l_{\T_1}(f_1)=l_{\T_2}(f_2)=1$ and
  $l_{\T_1}(f_2)=l_{\T_2}(f_1)=2$. Therefore,
  $l_{\T_1}(S)=l_{\T_2}(S)=3$, which means that MP will not favor either
  of the two trees $\T_1$, $\T_2$ over the other one. Moreover, as $f_1$
  and $f_2$ are incompatible with one another, it can easily be seen
  that both trees are actually MP-trees: the minimal score of either
  character is 1, as two states are employed, and this score is achieved
  when the character corresponds to a split on an edge of the underlying
  tree -- but because of the incompatibility, the other character will
  have a score of at least 2.  So for $S$, a score of 3 is best
  possible, and thus both $\T_1$ and $\T_2$ are MP-trees.  \par
  For ML, the situation is different. This is because the assignments of
  $f_1$ on $\T_2$ and $f_2$ on $\T_1$ differ, as highlighted by Figure
  \ref{bounded-fig-old}. In fact, character $f_1$ has a unique most
  parsimonious extension on $\T_2$, whereas $f_2$ has two most
  parsimonious extensions on $\T_1$. As we show in the following, for a
  sufficiently small upper bound $u$, the likelihood function is
  maximized when these extensions both contribute to the likelihood.  We
  use a symbolic algebra system to evaluate $P(f_i|\T, \bar{p})$ for
  $i=1,2$, for all trees on five taxa and at all corners of the
  feasibility region of $\bar{p}$ (see Lemma~\ref{corner}). More
  specifically, for the five-leaf-trees under investigation, there are
  seven edges to which either 0 or $u$ can be assigned, which gives
  $2^7=128$ possible parameter vectors $\bar{p}$ at which the likelihood
  might be maximized. We observe that $\max P(f_1|\T_1) =\max
  P(f_2|\T_2)=\frac{1}{2}u$, but $\max P(f_1|\T_2)=\frac{1}{2}u^2$ and
  $\max P(f_2|\T_1)= \max (\frac{1}{2}u^2, u^2(1-u)^2)$. So there are
  choices of $u$, namely all $u< 1-\frac{1}{\sqrt{2}}$, for which $\max
  P(f_1|\T_2) < \max P(f_2|\T_1)$. In these cases, even though both
  $\T_1$ and $\T_2$ are MP-trees, ML will favor tree $\T_1$ over
  $\T_2$. Therefore, MP and ML are not equivalent in this case.

  Now let sequence $\tilde{S}$ contain $n$ copies of character $f_1$ and
  $n+1$ copies of character $f_2$ for some integer $n >0$. Then, clearly
  $l_{\T_1}(\tilde{S})=3n+2$, but $l_{\T_2}(\tilde{S})=3n+1$. Therefore,
  MP will favor tree $\T_2$ over $\T_1$. Moreover, $\T_2$ is an MP-tree
  (by the same incompatibility argument concerning $f_1$ and $f_2$ as
  above). On the other hand, we have $\max P(\tilde{S}|\T_1) =
  (\frac{1}{2}u)^n \cdot \left(u^2(1-u)^2\right)^{n+1}$ and $\max
  P(\tilde{S}|\T_2) = \frac{u^{3n+1}}{2^{2n+1}}$ (provided $u<
  1-\frac{1}{\sqrt{2}}$). We choose a sufficiently large value of $n$ so
  that the former value is larger than the latter. For such choices of
  $n$, ML will favor tree $\T_1$ over $\T_2$, even though MP favors
  $\T_2$. It is important to note, however, that for the sequence
  $\tilde{S}$, the tree $\T_1$ is not an ML-tree. It can be easily
  verified for the tree $\T_3$ in Figure \ref{T_3} that $\max
  P(\tilde{S}|\T_3) = \left(\frac{u}{2}\right)^{n+1}
  \left(u^2(1-u)^2\right)^n$, which is more than $\max
  P(\tilde{S}|\T_1)$. In fact, $\max P(\tilde{S}|\T_3) > \max
  P(\tilde{S}|\T_1)$ for all $u \leq \frac{1}{2}$. In fact, further work
  shows $\T_3$ is the unique ML-tree. Moreover, $\T_3$ is also an
  MP-tree. So for $r=2$, it remains unclear whether MP and ML can make
  strictly conflicting choices.

  \begin{figure}[ht]
    \centering\vspace{0.5cm} 
    \includegraphics[width=6cm]{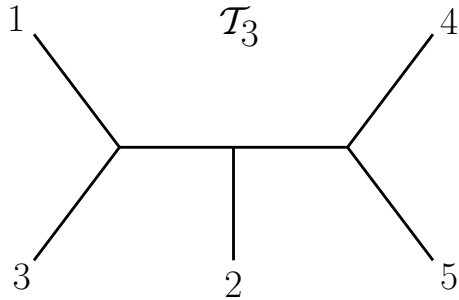} 
    \caption[Tree $\T_3$ is both an MP- and an ML-tree for sequence
    $\tilde{S}$] {Tree $\T_3$ is both an MP- and an ML-tree for sequence
      $\tilde{S}$}
    \label{T_3}
  \end{figure}

  Note that when $r=2$, examples demonstrating the inequivalence of MP
  and ML cannot be constructed with fewer than five taxa. This is because given at most one interior edge, it can be easily
  checked that all non-informative binary characters have the same
  maximum probability on all trees, whereas informative binary
  characters on four taxa have a higher probability on the tree where
  they have parsimony score 1 (calculation not shown).
 
  \noindent {\bf Case $r=3$:}

  Let the set of character states be $\{\alpha,\beta,\gamma \}$. We
  consider four taxa and the characters $f_1 := \alpha\alpha\beta\beta$
  and $f_2:=\alpha\beta\gamma\beta$, as well as the sequence $S$ of
  characters defined by $S:= f_1\underbrace{f_2 \ldots f_2}_{n
    \mbox{\scriptsize \hspace{0.02cm} times}}$. Two of the three
  possible trees on four taxa are shown in Figure \ref{4-taxa}: the tree
  $\T_4 = 12|34$ and the tree $\T_5=13|24$.

  \begin{figure}[ht] 
    \centering \subfloat[]{ \includegraphics[width=4cm]{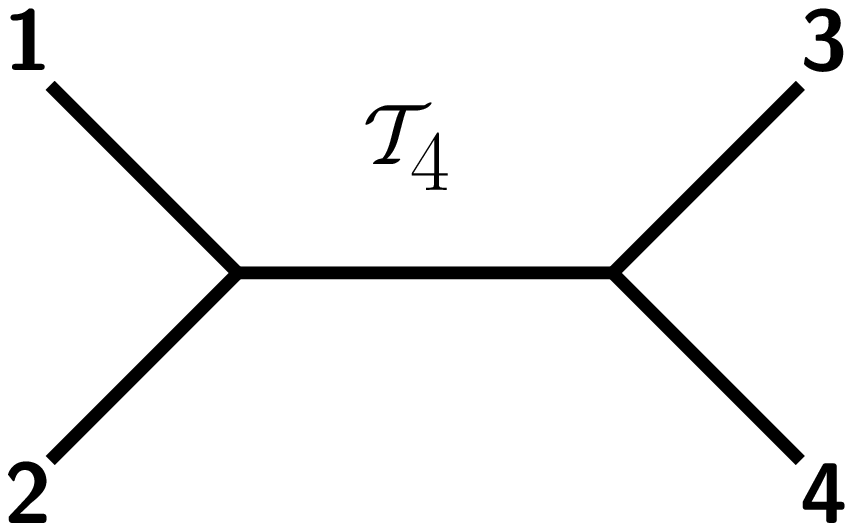}}
    \hspace{2.5cm}\subfloat[] {\includegraphics[width=4cm]{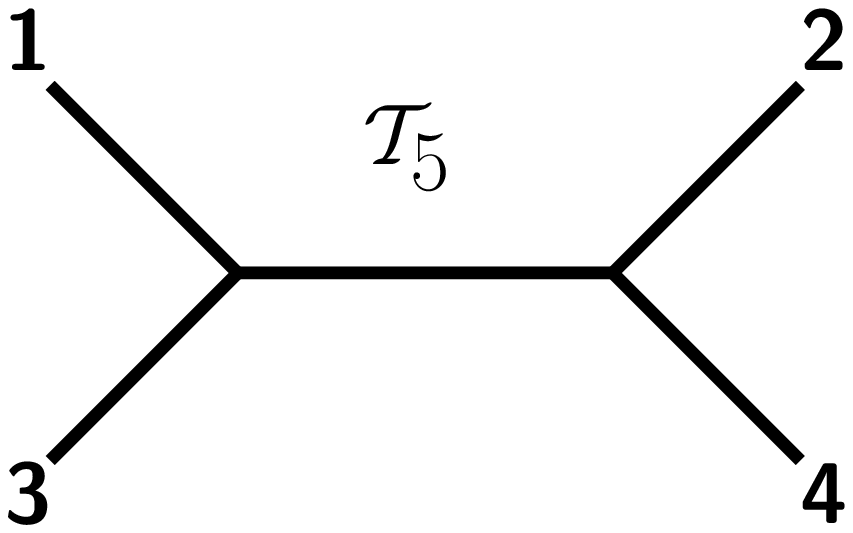} }
    \renewcommand*{\captionsize}{\footnotesize}
    \caption[Illustration of the unique MP tree versus the perfect
    distance-tree ]
    {Tree $\T_4$ illustrated in (a) is the unique MP-tree for $S$,
      whereas (b) depicts tree $\T_5$, which is the unique ML-tree for
      $S$ when $n$ is chosen sufficiently large.}
    \label{4-taxa}
  \end{figure}

  Tree $\T_4$ is clearly the unique MP-tree of $S$, as the only
  informative character in $S$ is $f_1=\alpha\alpha\beta\beta$.

  \par The ML-trees are obtained as described at the end of
  Section~\ref{notation}. As before, we used a symbolic algebra system
  to evaluate $P(f|\T, \bar{p})$ for all characters $f$ in the sequence,
  for all trees on four taxa and at all corners of the feasibility region
  of $\bar{p}$ (see Lemma~\ref{corner}). We observed that $\max
  P(f_2|\T_4)=\frac{u^2}{3}$ and $\max P(f_2|\T_5)
  =u^2(1-2u)$. Therefore, for all $u < \frac{1}{r}$, we have $\max
  P(f_2|\T_5) > \max P(f_2|\T_4)$. Now for any $u < \frac{1}{r}$, a
  sufficiently large value of $n$ may be chosen such that $\frac{\max
    P(S|\T_5)}{\max P(S|\T_4)} > 1$. We do not analyze the character
  $f_1$, although the actual choice of $n$ will depend on the ratio
  $\frac{\max P(f_1|\T_5)}{\max P(f_1|\T_4)}$ and on $u$. Therefore, MP and
  ML choose different trees in this three-state setting. Moreover, it turns
  out that for the third topology on four taxa, namely $\T_6 = 14|23$,
  we have $\max P(S|\T_6) < \max P(S|\T_5)$ for all choices of $u \leq
  \frac{1}{3}$ (calculations not shown). So, $\T_5$ is the unique
  ML-tree, whereas $\T_4$ is the unique MP-tree in this setting. So MP
  and ML make strictly conflicting choices.

  \noindent {\bf Case $r > 3$:}

  Let the set of states be $\C := \{\alpha, \beta, \gamma, \delta_1,
  \delta_2, \ldots, \delta_{r-3}\}$. Let $\D := \{\delta_1, \delta_2,
  \ldots, \delta_{r-3}\}$. We analyze four taxa and the same characters $f_1 :=
  \alpha\alpha\beta\beta$ and $f_2:=\alpha\beta\gamma\beta$ that were
  analyzed in the case $r=3$, but this time under the $N_r$-model with
  $r > 3$. Again we consider the sequence of characters $S:=
  f_1\underbrace{f_2 \ldots f_2}_{n \mbox{\scriptsize \hspace{0.02cm}
      times}}$.

  We only sketch the proof in this case. In particular, we indicate how
  the expressions for the likelihood function may be written regardless
  of the number of states.

  The expressions for $ P(f_i|\T_j)$ for $i=1,2$ and $j = 1,2,3$ can be
  written in a simple manner since the states $\delta_i$ do not occur in
  $S$. For example, let the substitution probabilities on the edges of a
  four-taxa tree $\T$ be $\bar{p} = (p_i, i = 1,2, \ldots, 5)$, where
  $p_i, i = 1,2,3,4$ are the substitution probabilities on the pending
  edges adjacent to taxa 1, 2, 3, 4, respectively, and $p_5$ is the
  substitution probability on the internal edge. Let $v$ and $w$ be the
  internal vertices of $\T$. We write $ P(f_i|\T, \bar{p}) = \sum_g
  P(g|\T,\bar{p})$, where the summation is over all extensions $g$ of
  $f_i$.

  Now observe that if $g$ and $h$ are two extensions of either $f_1$ or $f_2$, then we have $P(g|\T,\bar{p}) =
  P(h|\T,\bar{p})$ if $g(v), h(v) \in \D$ and $g(w) = h(w) = s \notin
  \D$ (or vice versa with the roles of $v$ and $w$ interchanged).

  Therefore:
  \begin{equation*}
    \sum_{g:g(v) \in \D, g(w)=s\notin \D }  P(g|\T,\bar{p}) 
    = (r-3)P(h|\T,\bar{p}),
  \end{equation*}
  where $h$ is an extension of $f$ for which $h(v) = \delta_1$ and $h(w) =
  s$. 

  Similarly:
  \begin{equation*}
    \sum_{g:g(v)=s \notin \D, g(w) \in \D }  P(g|\T,\bar{p}) 
    = (r-3)P(h|\T,\bar{p}),
  \end{equation*}
  where $h$ is an extension of $f$ for which $h(v) = s$ and $h(w) =
  \delta_1$.

  Finally: 
  \begin{equation*}
    \sum_{g:g(v) \in \D, g(w) \in \D }  P(g|\T,\bar{p}) 
    = (r-3)(1-3p_5)p_1p_2p_3p_4,
  \end{equation*}
  
  With these observations, it is possible to write the expressions
  for computing $P(f_i|T_j)$ in a computer algebra system. As in the
  case $r=3$, we analyzed only $P(f_2|\T_4)$ and $P(f_2|\T_5)$, and
  verified that $\max P(f_2|\T_5) \geq \frac{u^2(3-2ru)}{r}$ and $\max
  P(f_2|\T_4) = \frac{u^2}{r}$. Since $(3-2ru) > 1$ for all $u <
  \frac{1}{r}$, there is an $n$ for which $\max P(S|\T_5) > \max
  P(S|\T_4)$. This means that ML will favor $\T_5$ over $\T_4$, even
  though $\T_4$ is the unique MP-tree in this setting.
\end{proof}\par\vspace{0.3cm}

\begin{remark}
  \label{mike}
  It is important to state that in the examples for $r\geq 3$ introduced
  in the proof of Proposition \ref{no-clock}, where the number of taxa
  is bounded (in fact, it is only 4), as $u$ approaches $\frac{1}{r}$,
  we require $n$ to tend to infinity for ML and MP to make different
  choices.  However, this is a necessary property of any such example
  for which the number of taxa is bounded: For any fixed character
  sequence $S$, the continuity of the likelihood function and the
  Tuffley-Steel result (Theorem ~\ref{thm-tuffley-steel}) imply that
  there is a positive real number $\epsilon(S)$ such that if $u >
  \frac{1}{r} - \epsilon(S)$, then ML and MP choose the same sets of
  trees. Therefore, for a bounded number of taxa, since there are only
  finitely many sequences of length at most $k$, we set $\epsilon:=
  \min_S(\epsilon(S))$, where the minimization takes place over all
  character sequences of length at most $k$, and conclude that MP and ML
  would be equivalent (in the sense of the Tuffley-Steel result) for all
  $u > \frac{1}{r} - \epsilon$, for all sequences of length at most
  $k$. Therefore, as $u$ approaches $\frac{1}{r}$, the sequence length
  $n$ of sequences for which MP and ML make conflicting choices has to
  tend to infinity.
\end{remark}

We now complement the above inequivalence results by showing that for
sufficiently small choices of $u$, all ML-trees are also MP-trees. To
prove this result, we first establish lower and upper bounds for the
maximum probability of observing a character given a tree.

\begin{proposition}
\label{ml-bounds-2}
Let $\T$ be a phylogenetic $X$-tree, where $|X|=m$. Let $f$ be a
character on $X$. Let $0 \leq u < \frac{1}{r}$. Then under the
$N_r$-model with all substitution probabilities bounded by $u$, we have
\begin{equation*}
  \left(\frac{1}{r}\right) u^{l_{\T}(f)} \leq 
  \max P(f|\T) \leq r^{m-3}u^{l_{\T}(f)}.
\end{equation*}
\end{proposition}

\begin{proof}
  For the lower bound, just as in the Tuffley-Steel approach explained
  above, we take a most parsimonious extension $g$ of $f$ and assign
  substitution probability $u$ to each edge that has a substitution in
  $g$, and 0 to all other edges. Considering the $r$ possible root states
  (for an arbitrarily chosen root), this gives the lower bound for
  $\max_{\bar{p}} P(f|\T,{\bar{p}})$.

  To prove the upper bound, we observe that there are exactly $r^{m-2}$
  extensions of $f$, where $m-2$ is the number of internal vertices,
  each of which may be assigned any of the $r$ states. We will now
  analyze these extensions. Let $g$ be any extension of $f$. For substitution probabilities $p_e\in \{0,u\}$ assigned to the edges of
  the tree, the value of $P(g|\T, \bar{p})$ for an assignment of
  probabilities that maximizes $P(f|\T, \bar{p})$ is either 0 (if one of
  the edges where there is a substitution in $g$ has been assigned a
  substitution probability 0) or is given by
  \begin{equation} 
    P(g|\T,\bar{p}) 
    = \frac{1}{r}u^{k_1}(1-(r-1)u)^{k_2} 
    \leq \frac{1}{r}u^{k_1} \leq \frac{1}{r}u^{l_{\T}(f)}, 
  \end{equation}
  where $k_1 \geq l_{\T}(f)$ is the number of edges where there is a
  substitution in $g$, and $k_2$ is the number of edges which require no
  substitution in $g$ but have been assigned substitution probability
  $u$. The factor $\frac{1}{r}$ is caused by the $r$ different possible
  choices for the root state.
  
  The upper bound now follows by summing the probabilities of all
  extensions.
\end{proof}

Now we will use the above bounds to derive the desired conclusion on
ML-trees.

\begin{theorem}\label{MP_ML_ordering}
  Let $\T_a$ and $\T_b$ be two phylogenetic $X$-trees, where $|X|=m$,
  and let $S:=f_1,f_2,\ldots ,f_n$ be a sequence of characters on
  $X$. Let the substitution probabilities on all edges of $\T_a$ and
  $\T_b$ be bounded by $u \leq r^{(2-m)n}$. Then under the $N_r$-model
  with no common mechanism, we have:
\begin{equation*}
    l_{\T_b}(S)<l_{\T_a}(S) \Rightarrow \max P(S|\T_b) > \max P(S|\T_a).
  \end{equation*}
\end{theorem}
 
\begin{proof}
By Proposition \ref{ml-bounds-2}, we have:
\begin{equation} 
  \label{lower_bound_T_a}
  \max P(S|\T_a) \leq r^{(m-3)n}u^{\sum_i l_{\T_a}(f_i)}
\end{equation} 
\par\vspace{-1cm}
{\begin{center} and
\end{center}}
\par\vspace{-1cm}
\begin{equation} 
  \label{upper_bound_T_b} 
  \max P(S|\T_b) \geq \left(\frac{1}{r}\right)^nu^{\sum_i l_{\T_b}(f_i)}.
\end{equation} 
Note that for any positive integers $a$ and $b$ such that $b < a$ and
any positive constant $c$, for sufficiently small values of $u$, we have
$u^a < cu^b$. Now let $b:=\sum_i l_{\T_b}(f_i)$ and $a:= \sum_i
l_{\T_a}(f_i)$ and $c:=\frac{1}{r^{n+(m-3)n}} = r^{(2-m)n}$. Then
Equations (\ref{lower_bound_T_a}) and (\ref{upper_bound_T_b}) imply that
$\max P(S|\T_b)> \max P(S|\T_a)$.
\end{proof}

The following corollary directly follows from the above theorem.

\begin{corollary} \label{sufficiently_small_u} Let $S$ be sequence of
  $n$ characters on a set of $m$ taxa. Then there is an $\epsilon =
  \epsilon(m,n,r)$ such that under the $N_r$-model with no common
  mechanism and with all substitution probabilities subject to an upper
  bound $u \in [0,\epsilon)$, all ML-trees of $S$ are also MP-trees.
\end{corollary}

\subsection{Molecular clock} 

We now prove a statement similar to Proposition \ref{no-clock}, but with
substitution probabilities which conform to a molecular clock. Moreover,
we consider only the three-state symmetric model. Under the $N_3$-model,
we consider placing a bound $p_{max}$ on the probability of each
particular substitution from the root to any leaf. The value
$p_{max}=\frac{1}{3}$ means we place no bound beyond that already in the $N_3$-model, while $p_{max}<\frac{1}{3}$ limits the tree depth.

\begin{proposition} 
  \label{clock}
  Under the $N_3$-model with no common mechanism, with the substitution
  probabilities constrained by a molecular clock, MP and ML are not
  equivalent for any bound $p_{max} \in [0,\frac{1}{3}]$.
\end{proposition}

\begin{proof}

  Consider the two rooted four-taxa trees $\T_1$ and $\T_2$ along with
  substitution probabilities $p_i$ and $\tilde{p}_i$, respectively, on
  their edges as shown in Figure \ref{clock-fig}. The trees have the
  same shape but different leaf labels, and possibly different
  probabilities of a substitution from the root to any of its
  leaves. Under a molecular clock, we have $p_1 = p_2$ and $p_3 = p_4$
  in $\T_1$, and $\tilde{p}_1 = \tilde{p}_3$ and $\tilde{p}_2 =
  \tilde{p}_4$ in $\T_2$.  Let $p,\tilde{p} \in [0,p_{max}]$ be the
  probabilities of a substitution from the root $\rho$ to any leaf in
  $\T_1$ and $\T_2$, respectively.

  Then under the $N_3$-model, we write $p$ and $\tilde{p}$ in terms of
  the substitution probabilities on the edges of the trees as follows:
  \begin{eqnarray*}
    p
    &=& (1-2p_5)p_1 + p_5(1-2p_1) + p_5p_1 = p_1 + p_5 - 3p_1p_5 \\[1mm]
    &=& (1-2p_6)p_3 + p_6(1-2p_3) + p_6p_3 = p_3 + p_6 -
    3p_3p_6.
  \end{eqnarray*}
  Thus $p_5 = \frac{p-p_1}{1-3p_1}$ and $p_6 =
  \frac{p-p_3}{1-3p_3}$. Similarly, on $\T_2$, we have $\tilde{p}_5 =
  \frac{\tilde{p}-\tilde{p}_1}{1-3\tilde{p}_1}$ and $\tilde{p}_6 =
  \frac{\tilde{p}-\tilde{p}_2}{1-3\tilde{p}_2}$.

  \begin{figure}[ht]
    \centering
    \includegraphics[width=7cm]{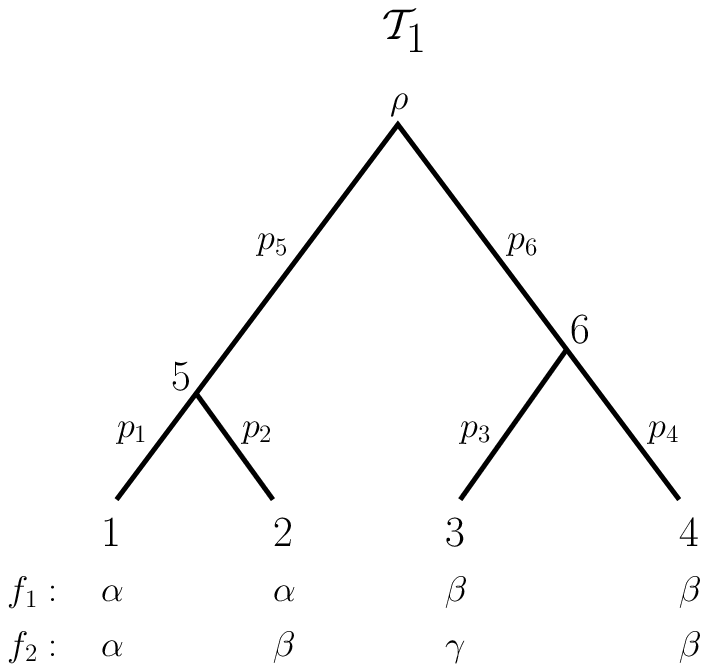}
    \hspace{1cm}
    \includegraphics[width=7cm]{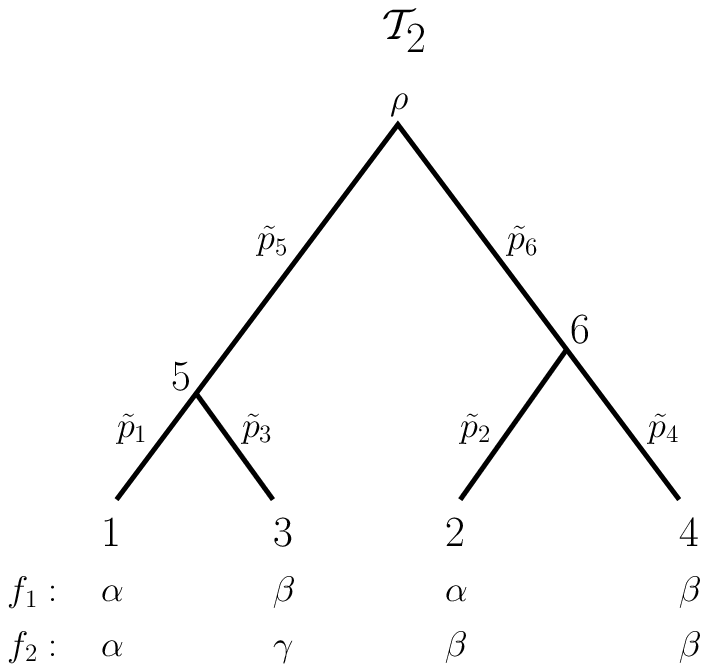}
    \caption[]{Rooted binary trees $\T_1$ and $\T_2$, which conform to a
      molecular clock, and the assignment of characters
      $f_1=\alpha\alpha\beta\beta$ and $f_2=\alpha\beta\gamma\beta$.}
    \label{clock-fig}
  \end{figure}

  As in the proof of Proposition~\ref{no-clock}, we consider the
  $N_3$-model with state space $\C := \{\alpha, \beta,
  \gamma\}$. Consider the characters $f_1 := \alpha\alpha\beta\beta$ and
  $f_2 := \alpha\beta\gamma\beta$, and a sequence of characters $S:=
  f_1\underbrace{f_2 \ldots f_2}_{n \mbox{\scriptsize \hspace{0.02cm}
      times}}$, where $n$ is a positive integer.

  As before, $\T_1$ is the unique MP-tree of $S$. We claim that $\T_1$
  is not an ML-tree if $n$ is sufficiently large. In order to show this,
  we show that $\max P(S|\T_2) > \max P(S|\T_1)$ for a suitable choice
  of $n$.

  We have
  \begin{equation*}
    \frac{ \max P(S|\T_2)}{ \max P(S|\T_1)} = 
    \frac{(\max P(f_1|\T_2))(\max P(f_2|\T_2))^n}
    {(\max P(f_1|\T_1))(\max P(f_2|\T_1))^n}.
  \end{equation*}

  We now demonstrate that $\max P(f_2|\T_2) > \max P(f_2|\T_1)$ for all
  values of $p_{max}$. This allows us to choose a sufficiently large
  value of $n$ so that the ratio above is more than 1.

  First we seek to maximize: $$P(f_2|\T_1, \bar{p}) = \sum_{c\in \C}
  P(f_2|\T_1, \bar{p}, \rho = c)P(\rho = c) = \frac{1}{3}\sum_{c\in
    \C}P(f_2|\T_1, \bar{p}, \rho = c).$$

  Using a computer algebra system, we expand the right-hand side of this equation by summing the probabilities over all possible
  assignments of states to the internal nodes 5 and 6, and substitute
  $p_5 = \frac{p-p_1}{1-3p_1}$ and $p_6 = \frac{p-p_3}{1-3p_3}$ to
  obtain:

  \begin{equation*}
    P(f_2|\T_1, \bar{p}) 
    = \frac{p_1p_3(3p_1p_3 - 2p_1 - 2p_3 + 1 + 2p -3p^2)}{3}.
  \end{equation*}

  Observe that for any fixed values of $p_1 + p_3$ and $p$, the expression above is maximized when $p_1p_3$ is maximized, i.e. when $p_1
  = p_3$. Therefore, we can substitute $p_1$ for $p_3$ and maximize the
  resulting expression given by:
  \begin{equation*}
    \frac{p_1^2(1 - p - p_1)(1 + 3p - 3p_1)}{3 }.
  \end{equation*}

  Under the constraint $p \in [0,p_{max}]$, straightforward arguments
  show that the expression shown above has a maximum at $p_1 = p =
  p_{max}$. Therefore:

  \begin{equation}
    \label{eq-f2t1}
    \max  P(f_2|\T_1) = \frac{p_{max}^2(1-2p_{max})}{3}.
  \end{equation}

  Similar calculations show that:

  \begin{equation}
    \label{eq-f2t2}
    \max P(f_2|\T_2) = p_{max}^2(1-2p_{max}),
  \end{equation}
  where the maximum is obtained by setting $\tilde{p}_2 = \tilde{p}_4 =
  \tilde{p}_5 = 0$ and $\tilde{p}_1 = \tilde{p}_3 = \tilde{p}_6 =
  \tilde{p} = p_{max}$.

  Equations ~(\ref{eq-f2t1}) and ~(\ref{eq-f2t2}) imply that $\max
  P(f_2|\T_2)> \max P(f_2|\T_1)$ for all
  $p_{max}\in(0,\frac{1}{3}]$. Now we can select a sufficiently large
  value of $n$ so that $ \max P(S|\T_1) < \max P(S|\T_2)$, where the
  actual choice of $n$ will depend on the ratio $\frac{\max
    P(f_1|\T_2)}{\max P(f_1|\T_1)}$ and $p_{max}$.
    
  This analysis does not show that $\T_2$ is an ML-tree, but it shows
  that $\T_1$, which is a unique MP-tree, is not an ML-tree. Therefore,
  the two methods are not equivalent under the constraint of a molecular
  clock, even when we assume no common mechanism.
\end{proof}

\par \vspace{0.3cm}

\section{Discussion and Outlook}

Our main objective was to present examples of sequences of characters
for which MP and ML with no common mechanism may choose different sets
of trees under the $N_r$-model when the substitution probabilities are
bounded above by $u < \frac{1}{r}$ or when a molecular clock is
assumed. Our four-taxa examples with $r \geq 3$ character states shows
that even if the upper bound $u$ is arbitrarily close to $\frac{1}{r}$,
we can find sequences of characters which are sufficiently long to cause
MP and ML to make conflicting choices.

The motivation for our four-taxa examples came from the idea of the
so-called `misleading sequences', which are sequences for which the
(parsimoniously) perfect phylogeny (i.e. a tree on which the whole
sequence is completely homoplasy-free) and the tree on which the derived
Hamming distances are additive differ (for details, see
\citep{huson_steel_2004}, \citep{bandelt_fischer_2008}). Even though
this discrepancy refers to perfect phylogenies (as opposed to general
MP-trees), we used a similar idea to construct our four-taxa
examples. In particular, the idea underlying the construction of our
sequences is based on the fact that MP ignores parsimoniously
non-informative characters in any sequence, whereas ML (just as
distance-based methods) does not. We exploited this fact to cause a
discrepancy between MP and ML by taking sufficient non-informative
characters.

It has been known that there are no binary `misleading sequences': if a
set of binary characters is convex on a binary phylogenetic tree, then
the Hamming distances of this sequence are a tree metric on the same
tree (see \citep{semple_steel_2003}, Prop. 7.1.9). But for MP and ML, it
is still unknown if there is a sequence of binary characters for which
these methods make conflicting choices when the substitution
probabilities are bounded above by $u < \frac{1}{r}$. We looked at
binary characters on five-taxa trees and found character sequences for
which some MP-trees are not ML-trees, but we observed that the ML-trees
in our examples were also MP-trees - which means that the equivalence of
MP and ML failed. But we did not find an example of strictly conflicting
choices in the binary case. Also, we did not find examples of sequences
for which MP and ML are not equivalent for values of $u$ which are
arbitrarily close to $\frac{1}{2}$. Thus, it would be interesting to
analyze two-state models further to decide if all ML-trees are MP-trees
and if the equivalence between MP and ML under no common mechanism can
fail for values of $u$ that are arbitrarily close to $\frac{1}{2}$.
  
MP is traditionally assumed to be justified (in the sense of agreement
with ML) whenever substitution probabilities are small (see, for
example, \citep{felsenstein_2004}). Therefore, our result, which shows
that an upper bound on the substitution probabilities can make the
equivalence of MP and ML fail under the $N_r$-model with no common
mechanism, is particularly surprising. On the other hand, we have shown
that for sufficiently small choices of the upper bound, all ML-trees are
at least also MP-trees (but not vice versa). So in summary, although MP
has been proven to agree with ML in the $N_r$-model under the assumption
of no common mechanism (and under no further constraints), our examples
show that this equivalence may fail when the model is changed
slightly. Therefore, we conclude that neither the presence nor the
absence of a common mechanism alone can justify MP in the sense of an
MP-ML equivalence. More research could be done on other models of
nucleotide substitution in order to analyze conditions under which ML
and MP may give conflicting results. This might highlight even more
differences between MP and ML.

\section{Acknowledgements}
We want to thank Mike Steel who suggested the topic of this paper and
made many useful comments (in particular Remark~\ref{mike}). We thank
the Allan Wilson Centre for Molecular Ecology and Evolution, and the
European Union grant HuBi (Hungarian Bioinformatics) at the Alfr\'{e}d
R\'{e}nyi Institute of Mathematics for funding. Bhalchandra Thatte is
currently supported by the EPSRC research grant ``From Population
Genomes to Global Pedigrees'' in Jotun Hein's group at Oxford. Finally,
we wish to thank two anonymous referees for making valuable suggestions
to improve the manuscript.

\bibliography{fischer_thatte_bibfile}
\end{document}